\documentclass[11pt]{article}
\usepackage[colorlinks=false]{hyperref}
\usepackage{amsfonts, amsmath, amssymb, amsthm}
\usepackage{enumitem, tikz, bm}

\usepackage{mathtools}
\DeclarePairedDelimiter{\ceil}{\lceil}{\rceil}

\definecolor{darkgray}{RGB}{64,64,64}
\definecolor{litegray}{RGB}{192,192,192}
\tikzstyle{block}=[draw, rectangle, minimum height=1cm, text width=2cm, text centered, draw=darkgray, font=\small]
\tikzstyle{block_medium}=[draw, rectangle, minimum height=1.5cm, text width=2cm, text centered, draw=darkgray, font=\small]
\tikzstyle{block_large}=[draw, rectangle, minimum height=2.5cm, text width=2cm, text centered, draw=darkgray, font=\small]
\tikzstyle{line} = [draw, -latex]

\author{
Zilin Jiang\thanks{Department of Mathematics, Massachusetts Institute of Technology, Cambridge, MA 02139, USA. Email: {\tt zilinj@mit.edu}. The work was done when Z. Jiang was a postdoctoral fellow at Technion -- Israel Institute of Technology, and was supported in part by the Israel Science Foundation (ISF) grant nos 1162/15, 936/16.}
\and Nikita Polyanskii\thanks{CDISE, Skolkovo Institute of Science and Technology, and Department of Mathematics, Technion -- Israel Institute of Technology. Email: {\tt nikita.polyansky@gmail.com}. Supported in part by ISF grant nos. 1162/15, 326/17 and the Russian Foundation for Basic Research (RFBR) through grant nos. 16-01-00440~A, 18-07-01427~A, 18-31-00310~MOL\_A.}
\and Ilya Vorobyev\thanks{CDISE, Skolkovo Institute of Science and Technology, and Moscow Institute of Physics and Technology. Email: {\tt vorobyev.i.v@yandex.ru}. Supported in part by RFBR through grant nos. 16-01-00440~A, 18-07-01427~A, \mbox{18-31-00361~MOL\_A}.}}

\title{On capacities of the two-user union channel with complete feedback}
\date{}

\newtheorem{theorem}{Theorem}
\newtheorem{lemma}[theorem]{Lemma}
\newtheorem{conjecture}{Conjecture}

\theoremstyle{definition}
\newtheorem{definition}{Definition}

\theoremstyle{remark}
\newtheorem{remark}{Remark}

\newcommand{\cC}{\mathcal{C}}
\newcommand{\cE}{\mathcal{E}}
\newcommand{\cO}{\mathcal{O}}
\newcommand{\cL}{\mathcal{L}}
\newcommand{\N}{\mathbb{N}}
\newcommand{\R}{\mathbb{R}}
\newcommand{\X}{\mathcal{X}}
\newcommand{\Y}{\mathcal{Y}}

\newcommand{\va}{\bm{a}}
\newcommand{\vb}{\bm{b}}
\newcommand{\vj}{\bm{j}}
\newcommand{\vx}{\bm{x}}

\renewcommand{\epsilon}{\varepsilon}
\newcommand{\al}{\alpha}
\newcommand{\ta}{\theta}
\newcommand{\la}{\lambda}

\newcommand{\pr}[1]{\operatorname{Pr}\left({#1}\right)}
\newcommand{\cpr}[2]{\operatorname{Pr}\left({#1}\mid{#2}\right)}

\newcommand{\abs}[1]{\left\lvert {#1}\right\rvert}
\newcommand{\dset}[2]{\left\{#1 : #2\right\}}
\newcommand{\sset}[1]{\left\{#1\right\}}
\newcommand{\from}{\colon}

\begin{document}

\maketitle

\begin{abstract}
	The exact values of the optimal symmetric rate point in the Cover--Leung capacity region of the two-user union channel with complete feedback were determined by Willems when the size of the input alphabet is 2, and by Vinck, Hoeks and Post when the size is at least 6. We complete this line of research when the size of the input alphabet is 3, 4 or 5. The proof hinges on the technical lemma that concerns the maximal joint entropy of two independent random variables in terms of their probability of equality. For the zero-error capacity region, using superposition coding, we provide a practical near-optimal communication scheme which improves all the previous explicit constructions.
\end{abstract}

\section{Introduction}

The \emph{two-user union channel}, first introduced in \cite{chang1981t} and rediscovered in \cite{vinck85}, is a discrete memoryless multiple-access channel\footnote{The terminology from information theory used throughout the article is standard, and can be found in \cite{cover2012elements}.}: the channel takes symbols $x_1, x_2$ from the input alphabet $\X := [q] = \sset{1,2,\dots, q}$ given by two senders, and outputs the union $y = \sset{x_1, x_2}$ from the output alphabet $\Y := \dset{y\subseteq [q]}{\abs{y} \in \sset{1, 2}}$. For the special case $q=2$, the union channel coincides with the \emph{two-user binary adder channel}.

Since a received $y\in \Y$ cannot be unambiguously decoded, the central problem in two-user communication theory is to coordinate the two senders to send simultaneously as much information as possible to a single receiver through $n$ uses of the union channel.

Let the message sets specified for the senders be of size $ M_1$ and $M_2$, and let $w_1 \in [M_1], w_2 \in [M_2]$ be two messages chosen by the two senders beforehand. During the $k$th use of the channel, two functions $e_{1k}$ and $e_{2k}$ respectively encode $w_1$ and $w_2$ to two codewords $x_{1k} \in [q]$ and $x_{2k} \in [q]$. The union channel then takes $x_{1k}, x_{2k}$ and outputs $y_k := \sset{x_{1k}, x_{2k}}\in \Y$. The sequence of outputs $(y_k)_{k=1}^n$ is decoded by the receiver to the estimate $(\hat{w}_1, \hat{w}_2)$ of $(w_1,w_2)$.

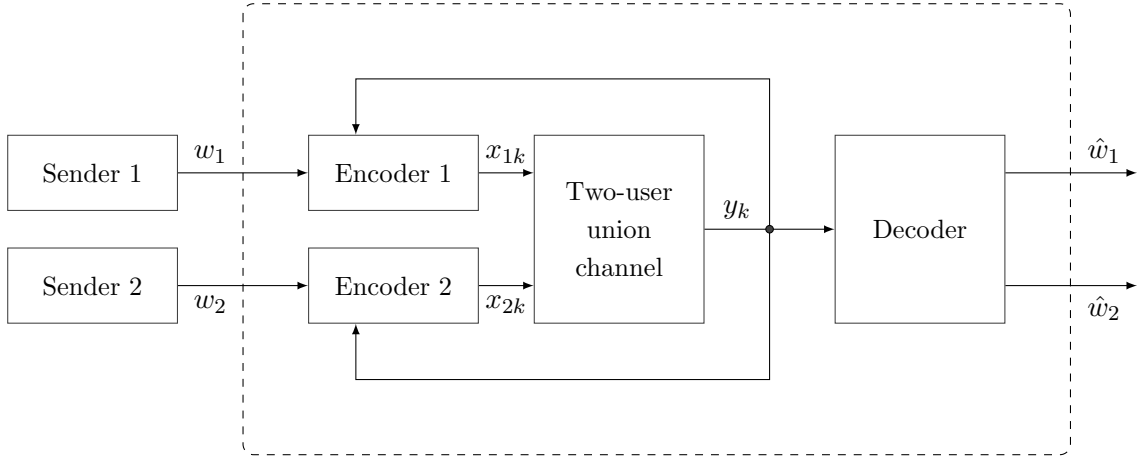
\begin{figure}[t]
  \centering
  \begin{tikzpicture}
    \node[block_large] (c) at (0,0) {Two-user union channel};
    \node[block_large] (d) at (4,0) {Decoder};
    \node[block] (e1) at (-3,0.75) {Encoder 1};
    \node[block] (e2) at (-3,-0.75) {Encoder 2};
    \node[block] (u1) at (-7,0.75) {Sender 1};
    \node[block] (u2) at (-7,-0.75) {Sender 2};
    \path[line] (u1) -- node[near start, above] {$w_1$} (e1);
    \path[line] (u2) -- node[near start, below] {$w_2$} (e2);
    \path[line] (e1.east) -- node[above] {$x_{1k}$} ([yshift=0.75cm]c.west);
    \path[line] (e2.east) -- node[below] {$x_{2k}$} ([yshift=-0.75cm]c.west);
    \path[line] (c) -- node[near start, above] {$y_k$} (d);
    \path[line] (2,0) -- (2,2) -- (-3.5,2) -- ([xshift=-0.5cm]e1.north);
    \path[line] (2,0) -- (2,-2) -- (-3.5,-2) -- ([xshift=-0.5cm]e2.south);
    \path[draw, dashed, rounded corners] (-5,3) -- (-5,-3) -- (6,-3) -- (6,3) -- cycle;
    \path[line] ([yshift=0.75cm]d.east) -- node[near end, above] {$\hat{w}_1$} +(1.75,0);
    \path[line] ([yshift=-0.75cm]d.east) -- node[near end, below] {$\hat{w}_2$} +(1.75,0);
    \node[draw, circle, minimum size=1mm, inner sep=0pt, outer sep=0pt, fill=darkgray] at (2,0) {};
  \end{tikzpicture}
  \caption{Two-user union channel with complete feedback.}
\end{figure}

In such a communication scheme for the transmission of information from two senders to one receiver, the encoders at any moment might know some information about the signals received by the decoder prior to this moment. When the encoders know nothing but the message from their corresponding senders, we say that the channel is the \emph{two-user union channel without feedback}. However, when the encoders also know all the previous outputs of the channel, namely every $e_{1k}$ and $e_{2k}$ depend not only on $w_1$ and $w_2$ respectively but also on $(y_i)_{i=1}^{k-1}$, we say that the channel is the \emph{two-user union channel with complete feedback}. See \cite[Section 4]{MR3699753} for a broader view on coding for the multiple-access channel.

In this work, we mainly focus on capacities of the two-user union channel with complete feedback. An $(M_1, M_2, n, \epsilon)$ code for the two-user union channel with complete feedback consists of a collection encoding functions and a decoding function such that the \emph{probability of error}, defined by $\pr{(\hat{w}_1, \hat{w}_2) \neq (w_1, w_2)}$ when $(w_1,w_2)$ is drawn uniformly from $[M_1]\times[M_2]$, is at most $\epsilon$. In particular, an $(M_1, M_2, n, 0)$ code could recover the messages without errors. The channel capacity region $\cE_f$ for the two-user union channel with complete feedback captures the rates at which the information can be transmitted over the channel for both users with arbitrarily small probability of error, whereas the zero-error capacity region $\cO_f$ represents the rates without error:
\begin{gather*}
  \cE_f := \text{closure of }\dset{\left(R_1,R_2\right)}{\exists\text{ a sequence of }(\ceil{q^{nR_1}}, \ceil{q^{nR_2}}, n, \epsilon_n)\text{ codes s.t. }\epsilon_n\to 0}, \\
  \cO_f := \text{closure of }\dset{(R_1,R_2)}{\exists\text{ a sequence of }(\ceil{q^{nR_1}}, \ceil{q^{nR_2}}, n, \epsilon_n)\text{ codes s.t. }\epsilon_n = 0\text{ eventually}}.
\end{gather*}
In the absence of feedback, the channel capacity region $\cE$ and the zero-error capacity region $\cO$ are similarly defined for the two-user union channel.

For each of the above capacity regions, say $\cC$, research has been devoted to the \emph{average capacity} \[
  R(\cC) := \sup\dset{\tfrac{1}{2}(R_1+R_2)}{(R_1,R_2)\in \cC},
\] which can be understood as the maximal rate per user at which the information can be transmitted. Because $\cC$ is convex and symmetric with respect to the line $R_1 = R_2$, the average capacity $R(\cC)$ can also be defined as $\sup\dset{R}{(R,R)\in\cC}$. The point $(R(\cC), R(\cC))$ is known as the equal-rate point or the symmetric rate point in the existing literature.

The channel capacity region for a discrete memoryless multiple-access channel without feedback has been fully characterized by Alswede~\cite{ahlswede1973multi} and Liao~\cite{liao1972multiple}. For the two-user union channel without feedback, the average channel capacity $R(\cE) = 1 - \frac{q-1}{2q\log_2q}$ has been determined by Chang and Wolf~\cite{chang1981t}. Much less is known for the average zero-error capacity $R(\cO)$ of the two-user union channel without feedback. For $q = 2$, there is no better upper bound other than the trivial $R(\cO)\le R(\cE) = 0.75$, while the current record lower bound is $R(\cO) \ge \frac{1}{12}\log_2 240 = 0.65891$ due to Mattas and \"Osterg{\aa}rd~\cite[Section III]{mattas2005new} obtained by computer searches. For $q\ge 3$, several constructions provided by Chang and Wolf~\cite[Section III]{chang1981t} imply that $R(\cO) \ge \frac{1}{4}\left(1+\log_q(q^2-q+1)\right)$ for all $q \ge 2$, $R(\cO) \ge \log_q\left(\frac{1}{2}(q+1)\right)$ for odd $q$ and $R(\cO) \ge \log_q\left(\frac{1}{2}\sqrt{q(q+2)}\right)$ for even $q$. The variation where the senders are required to use the same encoding functions was studied in various context. We refer the readers to \cite{lindstrom1969determination} for the best code construction when the size of the input alphabet is $2$, to \cite{cohen2001binary} for the connection with the binary $B_2$-sequences, and to \cite{gao2014new} for large input alphabet. The generalization, in which more than $2$ users have access to the channel, was recently investigated in \cite{MR3418936} and \cite{MR3724426}.

Gaarder and Wolf~\cite{gaarder75} demonstrated that feedback may increase the channel capacity region. They used the two-user binary adder channel as an example and developed a simple two-stage coding strategy. Using the concept of superposition coding Cover and Leung~\cite{cover1981achievable} characterized a subset of the channel capacity region $\cE_f$ for the discrete memoryless multiple-access channels with complete feedback. This subset was later shown to be exactly $\cE_f$ by Willems~\cite{willems1982feedback} for the class of the channels where one of the inputs is determined by the other input and the output. Their results are paraphrased as the following theorem in the special case that the channel is the two-user union channel.

\begin{theorem}[Theorem 1 of Cover and Leung~\cite{cover1981achievable} and Theorem of Willems~\cite{willems1982feedback}] \label{coverAchievable}
  The channel capacity region $\cE_f$ of the two-user union channel using input alphabet $[q]$ with complete feedback is the convex hull of all $(R_1,R_2)$ satisfying
  \begin{align*}
    0\le R_1 \le H(X_1\mid U), \quad
    0\le R_2 \le H(X_2\mid U), \quad
    R_1 + R_2 \le H(\sset{X_1, X_2})
  \end{align*}
  where $U$ is a discrete random variable\footnote{Salehi showed~\cite[Section III(d)]{salehi1978cardinality} that the channel capacity region is retained when the cardinality of $U$, denoted by $\abs{U}$, is restricted to ${q+1 \choose 2}$. This bound on $\abs{U}$ was also mentioned in \cite{cover1981achievable}. However, \cite{willems1982feedback} only referred to a slightly weaker bound $\abs{U} \le {q+1 \choose 2}+2$.}, $X_1, X_2$ are two $[q]$-valued random variables that are conditionally independent given $U$, and the entropy function $H$ uses the base-$q$ logarithm.
\end{theorem}

\begin{remark}
  The entropy $H(\sset{X_1, X_2})$ in Theorem~\ref{coverAchievable} is the entropy of the random variable $Y := \sset{X_1, X_2}$, not to be confused with the joint entropy $H(X_1, X_2)$.
\end{remark}

For $q=2$, Willems~\cite{willems1984multiple} later showed that $R(\cE_f) = 0.79113$; for $q\ge6$, Vinck, Hoeks and Post~\cite{vinck85} asserted that $R(\cE_f) = \frac{1}{2}\log_q{q+1 \choose 2}$. In Section~\ref{Shannon Capacity Feedback} we complete this line of research on $R(\cE_f)$ for all $q\ge2$. The proof hinges on the following lemma about the maximum of the joint entropy of two independent discrete random variables in terms of their probability of equality.

\begin{lemma}\label{joint entropy}
	Given $q \ge 2$, for every $\ta\in[0,1]$ let $F(\ta)$ be the maximum of the joint entropy $H(X_1, X_2)$ among all pairs of independent $[q]$-valued random variables $X_1, X_2$ such that $\pr{X_1 = X_2} = \ta$. The function $F\from [0,1]\to\R$ is continuous and it is increasing on $[0,1/q]$ and decreasing on $[1/q,1]$. Moreover,
	\[
    F(\ta) = 2\left(-\al \log\al - (1-\al)\log(1-\al) + (1-\al)\log(q-1)\right),\quad \text{for }\ta\in[1/q,1]
	\]
	where the bijection $\al\from [1/q,1]\to[1/q,1]$ is defined by
  \begin{equation}\label{def_alpha}
	  \al = \al(\ta) := \frac{1}{q} + \sqrt{\left(1-\frac{1}{q}\right)\left(\ta-\frac{1}{q}\right)}.
	\end{equation}
\end{lemma}

The proof of the lemma is provided in Appendix~\ref{proof lemma 3}. We mention that, given $X$ and the probability of equality, some optimization problems such as optimizing $H(Y)$ and minimizing $H(X,Y)$ were solved by Prelov in \cite{MR3273848} and \cite{MR3620431}.

As for the zero-error capacity region, Dueck~\cite[Section 2]{dueck1985zero} established a characterization for a class of discrete memoryless multiple-access channels including the two-user union channel with complete feedback. However, pinning down the precise value of $R(\cO_f)$ is still an open problem. For $q=2$, the best lower bound $R(\cO_f)\ge 0.78974$ was proved by Belokopytov~\cite{belokopytov1989zero} based on Dueck's characterization. Although Dueck's characterization shows that there exist good zero-error codes, it does not provide a way of constructing the best codes explicitly. If we use the scheme suggested by the proof of Dueck's theorem and generate a code at random with the appropriate distribution, the code constructed is likely to be good. However, without some structure in the code, it is computationally very difficult to decode. Hence the theorem does not provide a practical coding scheme.

In the context of group testing or a search problem on graphs, the ``Fibonaccian algorithm'' by Christen~\cite{christen1980fibonaccian} and Aigner~\cite{aigner1986search} gives an $\left(F_{n+1},F_n,n,0\right)$ code explicitly, where $F_n$ is the $n$th Fibonacci number, which implies that $R(\cO_f)\ge \log_2\phi = 0.69424$, where $\phi = 1.61834$ is the golden ratio. Later, the Fibonacci code was rediscovered by Zhang et al.~\cite[Theorem 1]{zhang1987some}, and was refined~\cite[Theorem 2]{zhang1987some} to achieve $R(\cO_f)\ge \log_2\phi' = 0.71662$, where $\phi' = 1.64333$ is the real root of $x^{11} = x^{10} + x^9 + 5$. Using the language of decision trees, Gargano et al.~\cite{gargano1992improved} constructs a $\left(32,32,7,0\right)$ code in an attempt to improve the Fibonacci code. Before our work, the best construction is a $\left(2^{235n+61},2^{235n+61},312n+123,0\right)$ code, for every $n\in\N$, due to Belokopytov and Luzgin~\cite{belokopytov1987block}, achieving $R(\cO_f) \ge 235/312 = 0.75321$. We present in Section~\ref{Zero Error Capacity Feedback} a practical communication scheme which achieves a near-optimal zero-error capacity for all $q$. For $q = 2$, our scheme achieves $R(\cO_f) \ge 0.77291$. For $q\ge 3$, our scheme is new and provides a lower bound that is close to the current upper bound.

We summarize the known results in Table~\ref{known values} for $q \le 6$ with our contribution in bold. We conclude in Section~\ref{Conclusion} with some open problems.

\begin{table}[t]
	\centering
	\begin{tabular}{c|c|c|c|c}
		\hline
		$q$ & $R(\cO)$ & $R(\cE)$ & $R(\cO_f)$ & $R(\cE_f)$ \\ \hline
		2 & [0.65891, 0.75] & 0.75 & [0.78974, 0.79113] & 0.79113 \\
		3 & [0.69281, 0.78969] & 0.78969 & [\textbf{0.81071}, \textbf{0.81510}] & \textbf{0.81510} \\
		4 & [0.71256, 0.8125] & 0.8125 & [\textbf{0.82946}, \textbf{0.83044}] & \textbf{0.83044} \\
		5 & [0.72292, 0.82773] & 0.82773 & [\textbf{0.84123}, \textbf{0.84130}] & \textbf{0.84130} \\
		6 & [0.72914, 0.83881] & 0.83881 & [\textbf{0.84953}, 0.84959] & 0.84959 \\ \hline
	\end{tabular}
  \caption{Summary of bounds on the average capacities of two-user union channels.}\label{known values}
\end{table}

\section{Channel capacity with complete feedback}\label{Shannon Capacity Feedback}

Hereafter $\log x$ stands for the base-$q$ logarithm of $x$. We define the entropy functions
\begin{equation*}
  H(x_1, \dots, x_k) := -\sum x_i \log x_i, \quad\text{for } x_i\ge 0 \text{ and } \sum x_i = 1.
\end{equation*}
and we abbreviate \[
  H\left(\underbrace{\frac{x_1}{r_1}, \dots, \frac{x_1}{r_1}}_{r_1}, \dots, \underbrace{\frac{x_k}{r_k}, \dots, \frac{x_k}{r_k}}_{r_k}\right)
\] by $H(x_1, \dots, x_k; r_1, \dots, r_k)$.

Under this notation, the function $F\from [0,1]\to\R$ defined in Lemma~\ref{joint entropy} can be written as
\begin{equation}\label{def_f}
  F(\ta) = 2H(\al,1-\al;1,q-1),\quad \text{for }\ta\in[1/q,1],
\end{equation}
where $\al = \al(\ta)$, as in \eqref{def_alpha}, is the larger root of the quadratic equation
\begin{equation}\label{def_al}
  (1-\al)^2 = (q-1)(\ta-\al^2) \quad\text{or equivalently}\quad q\al^2 - 2\al + 1 = (q-1)\ta.
\end{equation}
To express the average channel capacity $R(\cE_f)$, we need the concave envelope of $F$ and another function $G\from [0,1]\to\R$.

\begin{definition}
  The concave envelope of a continuous function $F\from I\to \R$ on a closed interval $I$, denoted by $\hat{F}$, is the lowest-valued concave function that overestimates or equals $F$ over $I$. It follows from the strengthened Carath\'{e}odory theorem by Fenchel and Eggleston~\cite[Theorem 18]{eggleston58} that for every $\ta\in I$, $\hat{F}$ is given by \[
    \hat{F}(\ta) = \max\dset{p_1F(\ta_1) + p_2F(\ta_2)}{0\le p_1, p_2\le 1, \ta_1, \ta_2\in I, p_1+p_2=1, p_1\ta_1+p_2\ta_2=\ta}.
  \]
\end{definition}

\begin{lemma}\label{technical lemma}
	The function $G\from[0,1]\to\R$ defined by
  \begin{equation}\label{def_g}
    G(\ta) := H\left(\ta, 1-\ta; q, \binom{q}{2}\right) = H(\ta, 1-\ta) + \ta + (1-\ta)\log\binom{q}{2}
  \end{equation}
  is concave and it attains its maximum $\log\binom{q+1}{2}$ at $2/(q+1)$.
\end{lemma}

\begin{proof}
  Since $H(\ta, 1-\ta)$ is concave and $G(\ta)-H(\ta,1-\ta)$ is a linear function of $\ta$, we conclude that $G(\ta)$ is also concave. Taking the derivative of $G$
  \[
    G'(\ta) = \log\left(\frac{1-\ta}{\ta}\right) + 1 - \log\binom{q}{2}
  \]
  and solving $G'(\ta) = 0$ yields the maximum point $\ta = 2/(q+1)$.
\end{proof}

\begin{theorem}\label{capacityFeedback}
  The average channel capacity $R(\cE_f)$ of the two-users union channel using input alphabet $[q]$ with complete feedback is given by
  \begin{equation}\label{eqn_fgt}
    R := \tfrac{1}{2}\max\dset{\min\left(\hat{F}(\ta),G(\ta)\right)}{\ta\in I},
  \end{equation}
  where $F$ and $G$ are defined by \eqref{def_f} and \eqref{def_g}, and $I := [1/q, 2/(q+1)]$. Moreover, the concave envelope of $F$ on $I$ is given by \[
    \hat{F}(\ta) = \begin{cases}
      F(\ta) & \text{when }q = 2; \\
      2-\frac{2(q-1)\log(q-1)}{q-2}\left(\ta-\tfrac{1}{q}\right) & \text{when }q \ge 3,
    \end{cases}
  \] and $R$ can thus be simplified to $\frac{1}{2}G(\ta)$, where $\ta$ is
  \begin{enumerate}[nosep]
    \item the solution of $\hat{F}(\ta) = G(\ta)$ in $I$, when $q = 2, 3, 4$;
    \item simply $\frac{2}{q+1}$, when $q \ge 5$.
  \end{enumerate}
\end{theorem}

\begin{proof}
	We first choose random variables $X_1, X_2, U$ in Theorem~\ref{coverAchievable} to demonstrate that $(R,R)$ is in the channel capacity region $\cE_f$. Let $\ta\in I$ be a maximizer of $\min(\hat{F}(\ta), G(\ta))$ in \eqref{eqn_fgt}, that is, $R$ is the minimum of $\frac{1}{2}\hat{F}(\ta)$ and $\frac{1}{2}G(\ta)$. By the definition of concave envelope, there are $p_1,p_2\in [0,1], \ta_1,\ta_2 \in [1/q,1]$ such that $p_1 + p_2 = 1$, $p_1\ta_1 +p_2\ta_2 = \ta$ and $p_1F(\ta_1) + p_2F(\ta_2) = \hat{F}(\ta) \le R$. Choose the discrete random variable $U$ as follows: for every $u\in\sset{1,2}, v\in [q]$, $U = (u, v)$ with probability $p_u/q$. Given $U = (u,v) \in \sset{1,2}\times[q]$, we choose two conditionally independent random variables $X_1, X_2$:
  \[
    X_1, X_2 = \begin{cases}
      v  &\text{with probability }\al(\ta_u), \\
      v' &\text{with probability }\frac{1-\al(\ta_u)}{q-1},\text{ for }v'\in [q]\setminus\sset{v}.
    \end{cases}
  \]
  Based on these choices of random variables, we obtain that for every $v\in [q]$,
  \begin{equation*}
    \pr{X_1 = X_2 = v} = \sum_{u=1}^2\frac{p_u}{q}\left(\al^2(\ta_u)+(q-1)\left(\frac{1-\al(\ta_u)}{q-1}\right)^2\right)\stackrel{\eqref{def_al}}{=}\sum_{u=1}^2\frac{p_u}{q}\ta_u=\frac{\ta}{q},
  \end{equation*}
  and for every $v_1 \neq v_2$,
  \begin{multline*}
    \pr{X_1 = v_1, X_2 = v_2} = \sum_{u=1}^2\frac{p_u}{q}\left(2\al(\ta_u)\frac{1-\al(\ta_u)}{q-1}+(q-2)\left(\frac{1-\al(\ta_u)}{q-1}\right)^2\right) \\
    \stackrel{\eqref{def_al}}{=}\sum_{u=1}^2\frac{p_u}{q}\frac{1-\ta_u}{q-1}=\frac{1-\ta}{q(q-1)}.
  \end{multline*}
  Then according to Theorem~\ref{coverAchievable}, the channel capacity region $\cE_f$ contains all $(R_1, R_2)$ satisfying,
	\begin{gather*}
	  R_i \le H(X_i\mid U) = \sum_{u=1}^2 p_uH(\al(\ta_u),1-\al(\ta_u);1,q-1) = \sum_{u=1}^2p_u \tfrac{1}{2}F(\ta_u) = \tfrac{1}{2}\hat{F}(\ta),\quad\text{for }i\in\sset{1,2},\\
	  R_1 + R_2 \le H(\sset{X_1, X_2}) = H\left(\ta,1-\ta; q, {q \choose 2}\right) = G(\ta).
	\end{gather*}
	Clearly $(R,R)$ satisfies these conditions, and so $(R, R)\in \cE_f$.

	Next we give a proof that $\cE_f$ is a subset of the half-space $\dset{(R_1,R_2)}{\frac{1}{2}(R_1+R_2) \le R}$. Since the half-space is already convex, from Theorem~\ref{coverAchievable}, it suffices to prove that
  \begin{equation}\label{eqn_subset}
    \min\left(H(X\mid U)+H(Y\mid U), H(\sset{X_1, X_2})\right) \le 2R,
  \end{equation}
  for every discrete random variable $U$ and $[q]$-valued random variables $X_1, X_2$ that are conditionally independent given $U$. Without loss of generality, we may assume that $U=u$ with probability $p_u$. Set
  \begin{align*}
    \ta^u_{\sset{v}} := \cpr{\sset{X_1,X_2}=\sset{v}}{U=u}, & \quad\text{for all }u \text{ and }v\in[q], \\
    \ta^u_{\sset{v_1,v_2}} := \cpr{\sset{X_1,X_2}=\sset{v_1,v_2}}{U=u}, & \quad\text{for all }u \text{ and }v_1, v_2\in[q], v_1 \neq v_2, \\
    \ta^u := \cpr{X_1 = X_2}{U = u}, & \quad\text{for all }u, \\
    \ta := \pr{X_1 = X_2} = \sum_u p_u\ta_u.
  \end{align*}
  On the one hand, as $X_1$ and $X_2$ are conditionally independent given $U$, by the definition of $F$ in Lemma~\ref{joint entropy}, we have
	\begin{equation}\label{eqn_onehand}
    H(X_1\mid U) + H(X_2\mid U) = \sum_u p_u H(X_1,X_2 \mid U=u)\le \sum_u p_u F(\ta^u) \le \hat{F}(\ta).
	\end{equation}
  On the other hand, because $x\mapsto -x\log x$ is concave, we obtain from Jensen's inequality that
	\begin{equation}\label{rates together}
	  H\left(\sset{X_1, X_2}\right) = H\left(\left(\sum_u p_u\ta^u_{\sset{v}}\right)_{v=1}^q, \left(\sum_u p_u \ta_{\sset{v_1,v_2}}^u\right)_{v_1\neq v_2}\right)\le
	H\left(\ta, 1-\ta;q,\binom{q}{2}\right) = G(\ta).
	\end{equation}
  Combining \eqref{eqn_onehand} and \eqref{rates together}, the left hand side of \eqref{eqn_subset} is at most
  \begin{equation}\label{eqn01}
    \min\left(\hat{F}(\ta), G(\ta)\right) \le \max\dset{\min\left(\hat{F}(\ta),G(\ta)\right)}{\ta\in[0,1]}.
  \end{equation}
  From Lemma~\ref{joint entropy}, $F$ is increasing on $[0,1/q]$ and decreasing on $[1/q,1]$, so is its concave envelope $\hat{F}$. Combining with Lemma~\ref{technical lemma} which says that $G$ is increasing on $[0,2/(q+1)]$ and decreasing on $[2/(q+1),1]$, we can find a maximizer of the right hand side of \eqref{eqn01} in $I$. Therefore the right hand side of \eqref{eqn01} equals $2R$.

  By the unimodality of $F$, we know that $\hat{F}$ restricted to $[1/q,1]$ is the same as the concave envelope of $F_0 := \left.F\right|_{[1/q,1]}$, the explicit formula of which is given by \eqref{def_f}. We are left to find a maximizer of $M\from I\to\R$ defined by $M(\theta) := \min\left(\hat{F_0}(\theta), G(\theta)\right)$. Since $\hat{F_0}$ is increasing on $I$, $G$ is decreasing on $I$, and $\hat{F_0}(1/q) > G(1/q)$, the maximizer of $M$ depends on which of $\hat{F_0}(2/(q+1))$ and $G(2/(q+1))$ is larger.

  \noindent\textbf{Case $q = 2$:} Observe that $F_0$ is concave already, and so $\hat{F_0} = F_0$. Since $F_0(2/(q+1)) < G(2/(q+1))$, the maximizer of $M$ is the solution of the equation $F(\theta) = G(\theta)$ in $I$.

  \noindent\textbf{Case $q \ge 3$:} Observe that $F_0$ has an inflection point $\ta^* \in (1/q,1)$, and $F_0$ is convex on $[1/q,\ta^*]$ and concave on $[\ta^*,1]$. Let $\ta'\in(\ta^*,1)$ be the point such that the line through the point $(1/q, F(1/q))$ and the point $(\ta', F(\ta'))$ is above the graph of $F$. In fact $\ta'$ is the root of the equation $$\frac{F(\ta')-F(1/q)}{\ta'-1/q} = F'(\ta')$$ in $(1/q,1)$, which turns out to be $\ta' = \frac{1}{q} + \frac{(q-2)^2}{q(q-1)}$. Define the linear function $L\from [1/q, \ta']\to\R$ to be
  \[
    L(\ta) = F(1/q) + F'(\ta')\left(\ta - \tfrac{1}{q}\right) = 2-\frac{2(q-1)\log(q-1)}{q-2}\left(\ta-\tfrac{1}{q}\right).
  \]
  Indeed the graph of $L$ is the line segment connecting $(1/q, F(1/q))$ and $(\ta', F(\ta'))$, and $\hat{F_0}(\ta) = L(\ta)$ for all $\ta\in(1/q,\ta')$. As $\ta' \ge 2/(q+1)$, we obtain that $$\hat{F_0}\left(\tfrac{2}{q+1}\right) = L\left(\tfrac{2}{q+1}\right) = 2 - \frac{2(q-1)^2\log(q-1)}{(q-2)q(q+1)}, \quad G\left(\tfrac{2}{q+1}\right) = \log\binom{q+1}{2}.$$
  It is enough to determine the sign of \[
    \left(\ln q\right)\left(\hat{F_0}\left(\tfrac{2}{q+1}\right) - G\left(\tfrac{2}{q+1}\right)\right) = \ln\left(\tfrac{2q}{q+1}\right) - \frac{2(q-1)^2}{(q-2)q(q+1)}\ln(q-1) =: \Delta(q).
  \]
  Compute directly $\Delta(3) = \ln\frac{3}{2}-\frac{2}{3}\ln2 < 0$ and $\Delta(4) = \ln\frac{8}{5}-\frac{9}{20}\ln3 < 0$, thus the maximizer of $M$ is the solution of the equation $L(\ta) = G(\ta)$ in $I$. In the cases $q \ge 5$, we estimate \[
    \Delta(q) = \ln 2 + \ln\left(1-\frac{1}{q+1}\right)-\left(1+\frac{1}{q(q-2)}\right)\frac{2\ln(q-1)}{q} \ge \ln2+\ln\left(1-\tfrac{1}{6}\right)-\left(1+\tfrac{1}{15}\right)\tfrac{2\ln 4}{5} > 0,
  \] thus the maximizer of $M$ is simply $\frac{2}{q+1}$ and the maximum of $M$ is $G(2/(q+1))$.
\end{proof}

\begin{remark}
	Observe that $R(\cE_f) \le \frac{1}{2}\log\binom{q+1}{2}$ is a naive bound since the output alphabet of the two-user union channel has cardinality $\binom{q+1}{2}$. In~\cite[Section II]{vinck85}, it was stated that equality holds in this naive bound for $q\ge 6$. It was then conjectured there that in Theorem~\ref{coverAchievable} a random variable $U$ of cardinality $q$ readily gives $(R,R)\in\cE_f$. However, the random variable $U$ used in our proof of Theorem~\ref{capacityFeedback} has cardinality $2q$, and we do not see a way to reduce its cardinality.
\end{remark}

\section{Zero-error capacity with complete feedback}\label{Zero Error Capacity Feedback}

In this section, we describe a zero-error communication scheme for the two-user union channel with complete feedback and show that it achieves a near-optimal rate pair $(R, R)\in \cO_f$. This scheme, like that in \cite[Section IV]{cover1981achievable}, partitions the uses of the channel into a large number $B+1$ of blocks, each of length $n$ except the last block. Suppose the message sets of the senders are both $[q]^{Bm}$, where $m\in[n]$ will be decided later, and let $w_1, w_2\in [q]^{Bm}$ be the messages of the two senders.

To describe the communication scheme in each block, we represent the uncertainty of the receiver about the first $bm$ digits of $w_1$ and $w_2$ at the end of block $b$ by $U(b) \subseteq ([q]\times[q])^{bm}$. In other words, the receiver, at the end of block $b$, knows that $$\left(\left(w_{1i}, w_{2i}\right)\right)_{i=1}^{bm} \in U(b).$$ The key idea of our communication scheme is to keep the uncertainty sets uniformly bounded in size.

Due to the feedback of the channel, the uncertainty set is common knowledge between the senders and the receiver. The initial uncertainty set $U(0) := \emptyset$. In addition, we assume for a moment that
\begin{equation} \label{inductive_assumption}
	\text{at the end of block } b \text{ each sender knows the first } bm \text{ digits of the other message.}
\end{equation}
This assumption will be shown below to hold by induction on the blocks.

\paragraph{Indexing:} At the start of block $b+1$, the senders and the receiver index the elements in $U(b)$ by \[
  S := \dset{(s_1, \dots, s_n)}{s_k\in\sset{*}\cup [q] \text{ such that }\abs{\dset{k}{s_k = *}}=m}.
\]
We shall choose $m\in[n]$ carefully in Theorem~\ref{construction feedback} so that $\abs{U(b)} \le \abs{S} = \binom{n}{m}q^{n-m}$. The method of indexing can be agreed beforehand between the senders and the receiver. For example, they can order both $U(b)$ and $S$ lexicographically, and index the elements in the ordered set $U(b)$ by the first $\abs{U(b)}$ elements in $S$.

\paragraph{Encoding:} During block $b + 1$, according to the inductive assumption \eqref{inductive_assumption} of our scheme, both senders know $\left(\left(w_{1i}, w_{2i}\right)\right)_{i=1}^{bm}$ and its index $(s_1, \dots, s_n) \in S$. During the $k$th use of the channel in block $b + 1$, both senders simply send $s_k$ if $s_k \in [q]$; and send $w_{1,bm+i}$ and $w_{2,bm+i}$ respectively if $s_k$ is the $i$th star in $(s_1, \dots, s_n)$. In the latter case, based on the feedback, each sender learns the $(bm+i)$th digit of the other message. Because there are a total of $m$ stars in $(s_1, \dots, s_n)$, at the end of block $b+1$, each sender learns $m$ more digits of the other message, maintaining the inductive assumption \eqref{inductive_assumption} of the scheme.

\paragraph{Decoding:} After $n$ uses of the channel in block $b+1$, the receiver has received $(y_1, \dots, y_n) \in (\Y_1 \cup \Y_2)^n$, where $\Y_i := \dset{y\subset[q]}{\abs{y}=i}$. The receiver then enumerates $(s_1, \dots, s_n)$ through the first $\abs{U(b)}$ elements in $S$, and for each $(s_1, \dots, s_n)$ compatible with the output, namely $y_k = \sset{s_k}$ if $s_k\in[q]$, the receiver adds to $U(b+1)$ all the $((\hat{w}_{1i},\hat{w}_{2i}))_{i=1}^{(b+1)m} \in ([q]\times[q])^{(b+1)m}$ such that
\begin{enumerate}[nosep]
  \item $((\hat{w}_{1i},\hat{w}_{2i}))_{i=1}^{bm} \in U(b)$ is indexed by $(s_1,\dots, s_n)$; and
  \item $\sset{\hat{w}_{1,bm+i},\hat{w}_{2,bm+i}} = y_{k_i}$ for all $i\in[m]$, where $k_1,\dots,k_m$ are the indices of $s_k$ that are stars.
\end{enumerate}
The updated uncertainty set $U(b+1)$ will be shown in Theorem~\ref{construction feedback} to be bounded by $\abs{S}$ in size.

\bigskip

After $B$ blocks of uses of the channel, the receiver obtains the uncertainty set $U(B)$, and both senders know $(w_1, w_2)$ and its index $(s_1, \dots, s_n)\in S$ as a member of $U(B)$. Finally, in the last block $B+1$, the senders simply communicate the index $(s_1, \dots, s_n)$ through $\ceil{\log \abs{S}} \le n - m + \ceil{\log \binom{n}{m}}$ uses of the channel.

\begin{theorem}\label{construction feedback}
  If $n/2 \le m \le n$ and \begin{equation}\label{assumption_m}
    \binom{2n-2m}{n-m}2^{2m-n} \le \binom{n}{m}q^{n-m},
  \end{equation} then the communication scheme described above allows the receiver to recover the messages $w_1, w_2 \in [q]^{Bm}$ from the senders without errors through $\le Bn + n - m + \ceil{\log\binom{n}{m}}$ uses of the two-user union channel with complete feedback. In particular, $R(\cO_f) \ge R$, where $R$ is the solution of $H_b(\al)+(1-\al)\log_2 q = 1$ in $(1/2,1]$, and $H_b$ is the binary entropy function.
\end{theorem}

\begin{proof}
  It suffices to show that $\abs{U(b)} \le \abs{S}$ for all $b = 0, 1, \dots, B$. The base case is evident as $U(0)$ consists of the empty sequence. For the inductive step, assume that $\abs{U(b)} \le \abs{S}$. During the $(b+1)$st block, the receiver has received $(y_1, \dots, y_n)\in(\Y_1\cup\Y_2)^n$.

  We shall estimate the size of the uncertainty set $U(b+1)$ at the end of the $(b+1)$st block. Suppose that $((\hat{w}_{1i},\hat{w}_{2i}))_{i=1}^{bm} \in U(b)$ is indexed by $(s_1, \dots, s_n)$. Recall that if $s_k \in [q]$, then $y_k = \sset{s_k}$; otherwise $s_k$ is the $i$th star and $\sset{\hat{w}_{1,bm+i}, \hat{w}_{2,bm+i}}=y_k$. In other words, only when $y_k = \sset{s_k}$ for every $k$ such that $s_k\in[q]$, the uncertainty set $U(b+1)$ would include the following $\prod_{i=1}^m{\abs{y_{k_i}}}$ elements \[
    \sset{((\hat{w}_{1i},\hat{w}_{2i}))_{i=1}^{(b+1)m}} \text{ such that } \sset{\hat{w}_{1,bm+i},\hat{w}_{2,bm+i}} = y_{k_i} \text{ for all }i\in[m],
  \] where $k_1, \dots, k_m$ are the indices of $s_k$ that are stars. Suppose $L := \dset{k}{y_k\in\Y_2}$ and $\ell := \abs{L}$. A $(s_1, \dots, s_n) \in S$ compatible with $(y_1, \dots, y_n)$ must have stars on coordinates indexed by $L$ and choose from the rest $n-\ell$ positions an additional $m-\ell$ stars. This $(s_1, \dots, s_n)$, if it indexes an element in $U(b)$, will contribute at most $2^\ell$ elements to $U(b+1)$. Therefore, we can estimate $\abs{U(b+1)} \le \binom{n-\ell}{m-\ell}2^\ell$.

  We claim that this estimate $u_\ell := \binom{n-\ell}{m-\ell}2^\ell$ reaches its maximum $\binom{2n-2m}{n-m}2^{2m-n}$ at $\ell = 2m-n, 2m-n+1$. In fact, we compare \[
    \frac{u_\ell}{u_{\ell+1}} = \frac{\binom{n-\ell}{m-\ell}2^\ell}{\binom{n-\ell-1}{m-\ell-1}2^{\ell+1}} = \frac{n-\ell}{2(m-\ell)},
  \] which is less than $1$ when $\ell < 2m-n$ and greater than $1$ when $\ell > 2m-n$. Combining with \eqref{assumption_m}, it is guaranteed that $\abs{U(b+1)} \le \abs{S}$. This finishes the inductive step.

  Finally we prove the lower bound on $R(\cO_f)$. Fix $\al\in(1/2,1]$ such that $H_b(\al) + (1-\al)\log_2q > 1$. It is well known that \[
    \lim_{n\to\infty}\frac{1}{n}\log_2\binom{n}{\ceil{\al n}} = H_b(\al).
  \]
  We can choose $n$ sufficiently large and $m = \ceil{\al n}$ so that \[
    \frac{1}{n}\log_2\binom{n}{m} + \left(1-\frac{m}{n}\right)\log_2q \ge 1 \implies \binom{n}{m}q^{n-m} \ge 2^n \ge \binom{2n-2m}{n-m}2^{2m-n}.
  \] Our communication scheme provides a $\left(q^{Bm}, q^{Bm}, Bn + n - m + \ceil{\log\binom{n}{m}}, 0\right)$ code. Thus \[
    R(\cO_f) \ge \frac{Bm}{Bn+n-m+\ceil{\log\binom{n}{m}}} \to \frac{m}{n} \ge \al \text{ as }B\to\infty.
  \] Note that $H_b(\al)+(1-\al)\log_2q - 1$ is decreasing on $(1/2,1]$ and has a unique root $R \in [0,1]$. We can choose $\al$ arbitrarily close to the root $R$ to show that $R(\cO_f) \ge R$.
\end{proof}

\begin{remark}
  For $q = 2$, our communication scheme could theoretically achieve $R(\cO_f) \ge 0.77291$. In practice, in order to achieve $R(\cO_f) \ge 0.764$, we can choose $m = 13, n = 17, B = 1019$ for our scheme to obtain a $\left(2^{13247},2^{13247},17339,0\right)$ code. During the encoding and decoding process, the senders and the receiver need to keep track of up to 35840 binary numbers of length 13247 in the uncertainty set, which takes up 59.35 megabytes of memory for storage.
\end{remark}

\begin{remark}
  For $q \ge 5$, a naive upper bound on $R(\cO_f)$ is $R(\cE_f) = \tfrac{1}{2}\log \binom{q+1}{2} = 1-\tfrac{1}{2\log_2 q}+O(1/q)$. Notice that when $\al = 1-\frac{1}{\log_2q}$, $H_b(\al) + (1-\al)\log_2q > (1-\al)\log_2q = 1$. The proof of Theorem~\ref{construction feedback} implies that $R(\cO_f)\ge 1-\frac{1}{\log_2q}$. Then gap between the upper bound and the lower bound on $R(\cO_f)$ is about $\frac{1}{2\log_2q}$.
\end{remark}

\section{Open problems}\label{Conclusion}

At the moment, the channel capacity $R(\cE_f)$ of the union channel with complete feedback is determined for every $q \ge 2$. The zero-error capacity $R(\cO_f)$ however is yet to be determined for any $q$. Naturally, the first step for future research is to determine $R(\cO_f)$ for $q = 2$. Based on the characterization of $\cO_f$ by Dueck~\cite{dueck1985zero}, our numerical experiments suggest that that for the binary case the lower bound on $R(\cO_f)$ proved by Belokopytov~\cite{belokopytov1989zero} is tight.

\begin{conjecture}
	The average zero-error capacity of the two-user binary adder channel, that is the union channel with $q = 2$, with complete feedback is equal to $0.78974$.
\end{conjecture}

An inspection of Table~\ref{known values} reveals that $R(\cE)<R(\cO_f)$ for $q\le 6$. However, for $q\ge 14$, the lower bound on $R(\cO_f)$ in Theorem~\ref{construction feedback} is less than $R(\cE)$. We speculate that our lower bound on $R(\cO_f)$ can be improved for every $q\ge 2$.

\begin{conjecture}\label{conjb}
  For all $q\ge 2$, the average channel capacity $R(\cE)$ of the two-user union channel without feedback is strictly less than the average zero-error capacity $R(\cO_f)$ of the two-user union channel with complete feedback.
\end{conjecture}

Conjecture~\ref{conjb} would establish the chain of inequalities $R(\cO) \le R(\cE) < R(\cO_f) \le R(\cE_f)$.

\section*{Acknowledgements}

We thank the referees for the suggestions which improved both the exposition of the paper and the clarity of the proofs.

\appendix

\section{Proof of Lemma~\ref{joint entropy}}\label{proof lemma 3}

\begin{proof}[Proof of Lemma~\ref{joint entropy}]
  Let $X_1$ and $X_2$ be independent $[q]$-valued random variables so that $\pr{X_1=X_2} = \ta$, and let $a_i := \pr{X_1 = i}$ and $b_i := \pr{X_2 = i}$ for every $i\in[q]$. Then $F(\ta)$ is
  \begin{subequations}\label{opt}
    \begin{align}
      \text{the maximum of}\quad & H(a_1, \dots, a_q) + H(b_1, \dots, b_q), \\
      \text{subject to}\quad & \sum a_i = 1, \sum b_i = 1, \sum a_ib_i = \ta, a_i, b_i \ge 0. \label{optc}
    \end{align}
  \end{subequations}
  Clearly $F(\ta)$ is continuous with respect to $\ta$.

  We first prove the unimodality of $F$, that is, for any $\ta \in [0,1]$, if $\ta'$ is between $\ta$ and $1/q$, then $F(\ta') \ge F(\ta)$. Let $\va = (a_1, \dots, a_q), \vb = (b_1, \dots, b_q)$ be a maximizer of the optimization problem \eqref{opt}. Let $\vj = (1/q, \dots, 1/q) \in \R^q$, and consider $\va' = (1-t)\va + t\vj$ and $\vb' = (1-t)\va + t\vj$, where $t\in[0,1]$ is a solution of
  \begin{equation}\label{ivt}
    \ta' = \va'\cdot\vb' = ((1-t)\va+t\vj)\cdot((1-t)\vb+t\vj) = (1-t)^2\ta + \frac{(2-t)t}{q}.
  \end{equation}
  The right hand side of \eqref{ivt} equals $\ta$ and $1/q$ when $t = 0, 1$ respectively. By the intermediate value theorem, \eqref{ivt} has a solution in $[0,1]$ and $t$ is well-defined. Note that $\va', \vb'$ satisfy the constraint \eqref{optc}. By the concavity of the entropy function, we have $$F(\ta') \ge H(\va')+H(\vb') \ge (1-t)H(\va)+tH(\vj)+(1-t)H(\vb)+tH(\vj) \ge H(\va) + H(\vb) = F(\ta).$$

  From this point forward, we assume that $\ta \in (1/q,1]$ is fixed. We shall repeatedly add constraints to the optimization problem \eqref{opt} without decreasing its maximum.

  Given $\va, \vb$ satisfying the constraint \eqref{optc}, consider the vectors $\va', \vb'$ whose coordinates are respectively the ones of $\va, \vb$ sorted in non-decreasing order. The rearrangement inequality says
  \[
    \ta' := \left(\frac{\va'+\vb'}{2}\right)\cdot\left(\frac{\va'+\vb'}{2}\right) \ge \va'\cdot \vb' \ge \va\cdot\vb = \ta.
  \]
  As $\ta$ is between $1/q$ and $\ta'$, again we have $\va''=\vb''=(1-t)\frac{\va'+\vb'}{2}+t\vj$ such that $\va''\cdot\vb''=\ta$ for some $t \in [0,1]$. By the concavity of the entropy function, we have
  \begin{multline*}
    H(\va'') + H(\vb'') \ge 2\left((1-t)H\left(\frac{\va'+\vb'}{2}\right)+tH(\vj)\right) \\
    \ge 2H\left(\frac{\va'+\vb'}{2}\right) \ge H(\va') + H(\vb') = H(\va)+H(\vb).
  \end{multline*}
  We come to the conclusion that without loss of generality we may assume $\va = \vb$ in \eqref{opt}. The optimization problem is then equivalent to
  \begin{subequations}
    \begin{align} \label{opt2}
      \text{Maximize:}\quad & \tfrac{2}{\ln q}\left(-\sum_{i=1}^q x_i\ln x_i\right),\\
      \text{Subject to:}\quad & \sum_{i=1}^q x_i = 1, \sum_{i=1}^q x_i^2 = \ta, x_i \ge 0 \text{ for all }i\in[q]. \label{opt2c}
    \end{align}
  \end{subequations}
  Consider the Lagrangian
  \[
    \cL(\vx, \la_1, \la_2) = -\sum_{i=1}^q x_i\ln x_i + \la_1\left(\sum_{i=1}^q x_i - 1\right) + \la_2\left(\sum_{i=1}^q x_i^2 - \ta\right).
  \]
  The method of Lagrange multipliers gives necessary conditions for the maximizers:
  \[
    \frac{\partial \cL}{\partial x_i} = -\ln x_i - 1 + \la_1 + 2\la_2 x_i = 0, \quad \text{for all }i\in[q].
  \]
  In other words, given $\la_1, \la_2$, each coordinate of a maximizer $\vx$ is a solution of the equation $-\ln x - 1 + \la_1 + 2\la_2 x = 0$. Since $x\mapsto -\ln x - 1 + \la_1 + 2\la_2x$ is convex, there are at most two solutions. Without loss of generality, we can add to \eqref{opt2c} the constraints that
  \[
    x_1=\dots=x_r=a,\quad x_{r+1}=\dots=x_q=b,
  \]
  for some $r \in [q-1]$, and $a \ge b \ge 0$. We have thus reduced the optimization problem \eqref{opt2} to
  \begin{subequations} \label{opt3}
    \begin{align}
      \text{Maximize:}\quad & \tfrac{2}{\ln q}\left(-ra\ln a - sb\ln b\right),\\
      \text{Subject to:}\quad & ra+sb=1, ra^2 + sb^2=\ta, r+s=q, r,s \in [q], a\ge b\ge 0.
    \end{align}
  \end{subequations}
  Given $r, s$ such that $r+s=q$, we can solve from $ra+sb=1, ra^2+sb^2=\ta$ for $a,b$:
  \[
    (a,b) = \left(\frac{1}{q}+\frac{\sqrt{rsp}}{qr}, \frac{1}{q}-\frac{\sqrt{rsp}}{qs}\right) \text { or }\left(\frac{1}{q}-\frac{\sqrt{rsp}}{qr}, \frac{1}{q}+\frac{\sqrt{rsp}}{qs}\right).
  \]
  where $p := \ta q -1 \in (0,q-1]$. Because $a \ge b$, we discard the second solution of $(a,b)$. Given the parameter $t := r/q \in [1/q,1-1/q]$, the variables $r = qt, s = q - qt$ and \[
    a = \frac{1}{q}\left(1 + \sqrt{\frac{1-t}{t}p}\right),\quad b= \frac{1}{q}\left(1 - \sqrt{\frac{t}{1-t}p}\right)
  \]
  can be seen as functions of $t$, so can the objective function $v(t) := -ra\ln a-sb\ln b$. The derivative of $v$ is
  \[
    v'=-q\left(a\ln a + ta'(\ln a + 1) - b\ln b +(1-t)b'(\ln b + 1)\right).
  \]
  The implicit differentiation of $ra+sb=1, ra^2+sb^2=\ta$ yields
  \[
    q\left(a + ta' - b + (1-t)b'\right) = 0,\quad a^2 + 2taa' - b^2 + 2(1-t)bb' = 0.
  \]
  which can be viewed as a system of linear equations of $a', b'$. Using $a > 1/q > b$, we deduce that
  \[
    a' = -\frac{a-b}{2t}, \quad b' = -\frac{a-b}{2(1-t)}.
  \]
  and we simplify $v'$ as follows:
  \begin{multline*}
    v' = -q\left(a\ln a - \tfrac{1}{2}(a-b)(\ln a +1 )-b \ln b - \tfrac{1}{2} (a-b)(\ln b + 1)\right) \\
    = -q\left(\frac{a+b}{2}\ln\left(\frac{a}{b}\right)+ b-a\right)=qb\left(-\tfrac{1}{2}(r+1)\ln r-1+r\right),
  \end{multline*}
  where $r = a/b > 1$. One can check that $-\tfrac{1}{2}(r+1)\ln r-1+r < 0$ for $r > 1$, and so $v(t)$ is decreasing. In other words, the objective function $v$ attains its maximum at $t = 1/q$ or equivalently the maximizer of \eqref{opt3} is given by $r = 1, s = q-1, a = \al, b = \frac{1-\al}{q-1}$, which leads to $F(\ta) = 2H(\al, 1-\al; 1, q-1)$.
\end{proof}

\bibliographystyle{alpha}
\bibliography{capacity}
\end{document}